\documentclass[conference]{IEEEtran}
\IEEEoverridecommandlockouts
\usepackage{cite}
\usepackage{amsmath,amssymb,amsfonts}
\usepackage{amsthm}
\newtheorem{theorem}{Theorem}

\newtheorem{lemma}{Lemma}

\theoremstyle{definition}
\newtheorem{definition}{Definition}
\theoremstyle{remark}
\newtheorem{remark}{Remark}
\newcommand{\sem}[1]{[\![#1]\!]}
\usepackage{algorithmic}
\usepackage{graphicx}
\usepackage{booktabs}
\usepackage{threeparttable}
\usepackage{textcomp}
\usepackage{xcolor}
\usepackage{url}
\usepackage[normalem]{ulem}
\usepackage[most]{tcolorbox}

\definecolor{brandblue}{RGB}{60,120,200}
\tcbset{
  mybox/.style={
    enhanced,
    colback=brandblue!4,
    colframe=brandblue!60!black,
    boxrule=0.6pt,
    arc=1mm,
    left=2mm,
    right=2mm,
    top=1mm,
    bottom=1mm,
    fonttitle=\bfseries,
    coltitle=black,
    attach boxed title to top left={xshift=0mm,yshift*=-\tcboxedtitleheight/2},
    boxed title style={
      colback=white,
      colframe=brandblue!60!black,
      boxrule=0.6pt,
      arc=1mm,
      width=\linewidth,
    },
  }
}

\usepackage{listings}
\lstdefinestyle{cdiff}{
  language=C,
  basicstyle=\ttfamily\scriptsize,
  keywordstyle=\color{blue!65!black},
  commentstyle=\color{green!45!black}\itshape,
  numbers=none,
  showstringspaces=false,
  columns=fullflexible,
  keepspaces=true,
  breaklines=true,
  frame=none,
  xleftmargin=0pt,
  aboveskip=1pt,
  belowskip=1pt,
}
\def\BibTeX{{\rm B\kern-.05em{\sc i\kern-.025em b}\kern-.08em
T\kern-.1667em\lower.7ex\hbox{E}\kern-.125emX}}
\usepackage[colorlinks=true,linkcolor=blue,citecolor=blue,urlcolor=blue,
  pdftitle={Partial Contracts Suffice: Sound, LLM-Inferred Regression
Verification}]{hyperref}

\newif\ifanonymized
\anonymizedfalse

\newcommand{\ourtool}[0]{\textsc{Contractor}}

\begin{document}

\title{Partial Contracts Suffice: Sound, LLM-Inferred Regression
Verification}

\ifanonymized{}
\author{\IEEEauthorblockN{Anonymous Author(s)}}
\else{}
\author{
  \IEEEauthorblockN{Yiannis Charalambous}
  \IEEEauthorblockA{\textit{The University of Manchester} \\
  0009-0000-5755-5099}
  \and
  \IEEEauthorblockN{Rafael Menezes}
  \IEEEauthorblockA{\textit{The University of Manchester} \\
  0000-0002-6102-4343}
  \and
  \IEEEauthorblockN{Youcheng Sun}
  \IEEEauthorblockA{\textit{The University of Manchester} \\
  0000-0002-1893-6259}
  \and
  \IEEEauthorblockN{Lucas C. Cordeiro}
  \IEEEauthorblockA{\textit{The University of Manchester} \\
  0000-0002-6235-4272}
}
\fi

\maketitle

\begin{abstract}
  Software evolves continuously, yet ensuring that a patch preserves
  intended behavior without re-verifying an entire codebase remains
  difficult. \emph{Regression verification} addresses this problem,
  but existing techniques either require expensive whole-program
  reasoning or rely on manually written specifications that are
  rarely available in practice.
  In this work, we present the first contract-based regression
  verification tool. In our approach, the contract soundness is
  ensured by proving that all function versions match the behavior.
  Additionally, to ensure the behavior, the contract is used to
  verify the program flow (i.e., assume-guarantee). We ask whether a
  partial, caller-sufficient contract, rather than a full behavioral
  specification, is enough. On Frama-C-Problems we strengthen each
  inferred contract past what the caller needs and measure how much
  tighter it becomes. It barely moves: for most targets in every
  model the caller-sufficient contract is already the tightest the
  loop reaches, and our tightness comparator rates the partial and
  strengthened contracts equivalent for the large majority of targets
  it can compare (18 of 20, 13 of 15, and 14 of 18 across the three
  models). Partial-spec contracts thus capture nearly all the
  attainable tightness, so stopping at caller-sufficiency costs
  almost nothing. The regression check underneath is sound: on the
  third-party EqBench-C suite it never fabricates an equivalence,
  returning zero false proofs on the pairs it decides and reporting
  an unprovable difference rather than a false proof. Used as a
  soundness oracle this way, it also surfaced nine pairs that EqBench
  mislabels as equivalent, more than a concurrent tool reports. The
  contracts themselves are inferred automatically from the checker's
  own counterexamples, with no separate specification step; on
  Frama-C-Problems and the ANSSI X509 parser this reaches a
  verification rate comparable to the weakest-precondition (WP) tools
  AutoSpec and Preguss, while a passing result certifies at least as
  strong a property, which we call \emph{safety-preserving
  conditional equivalence}: enforcement plus caller-sufficiency, with
  runtime-error freedom covered by the safety encoding.
\end{abstract}


\section{Introduction}\label{sec:intro}

Software changes constantly, and most changes are meant to be
conservative: a refactor, a bug fix, or a small feature should preserve
the behavior callers already rely on. They don't always manage it. A fix
that corrects one execution path can quietly perturb an output some
other caller reads, and if no test exercises that caller the regression
ships unnoticed; across major operating systems, between 14.8\% and
24.4\% of post-release fixes were themselves found to be
incorrect~\cite{yin2011fixes}.
Catching it means deciding whether the new and old versions still
agree, and there's no general procedure for that: deciding full
input/output equivalence is undecidable~\cite{rice1953}, and even a
bounded check stays expensive, since it has to reason about each
function together with all its side-effects~\cite{godlin2013rv}.

In a bug-fixing or refactoring context, full equivalence asks too much:
the developer means to change some paths and leave the rest, so
requiring agreement everywhere, old bugs and all, rejects the very edit
being made. The useful notions form a spectrum, ordered by how much of
each execution they compare, and the weaker the relation the cheaper the
proof. \emph{Full equivalence} is the strongest: the two versions must
terminate on the same inputs and return the same outputs everywhere,
preserved bugs included. \emph{Partial equivalence} relaxes termination,
asking the versions to agree only on the inputs where the old one
terminates. \emph{Conditional}, or \emph{regression}, \emph{equivalence}
relaxes further still: it fixes a property marking the behavior worth
preserving and asks for agreement only where the old version already
satisfied it. Where the old version violated the property, the new one
is free to differ, and that freedom is exactly where a bug fix lives. We
make the three notions precise in Definition~\ref{def:equiv}.

Whichever notion one targets, a whole-program check re-analyzes every
callee in full at each call site, which does not scale as a codebase
grows. A scalable proof has to be \emph{modular}: each call is
discharged against an \emph{abstraction} $\alpha$ of the callee rather
than its body, so the callee is analyzed once instead of at every use.
For the proof to stay trustworthy that abstraction must be \emph{sound}:
it has to \emph{over-approximate} the callee, admitting every behavior
the function can exhibit, so the check can never miss a real difference.

We take that abstraction to be a \emph{function contract}: a precondition
and a postcondition that pin down just part of the callee's behavior
(Definition~\ref{def:contract}). A verifier \emph{enforces} the contract against
the old version, certifying it as a sound over-approximation, and then
\emph{replaces} each call by it, so one enforced contract serves every
call site without re-entering the body. The contract need only be precise
on the slice of the output the property observes; off that slice it may
stay weak, and looseness there costs completeness, never soundness.
Anchoring it to the old version's safe behavior also folds safety in, so a
newly introduced memory leak or out-of-bounds access counts as a
regression even when every returned value is unchanged.

Contracts like this are almost never written by hand: even
safety-critical code rarely carries machine-checkable per-function
specifications, and what exists tends to be
system-level~\cite{wen2024autospec,granberry2024specifywhat,wang2025preguss}.
We therefore \emph{infer} them with a large language model, which has
recently begun to automate contract authoring, though so far only for a
single version of a
program~\cite{wen2024autospec,granberry2024specifywhat,wang2025preguss}.
The model is never trusted: it \emph{proposes} a contract, the verifier
\emph{enforces} it against the old version, and only an enforced
contract is used, so a wrong guess is rejected rather than believed.
This is the paper's thesis: \textit{a partial, LLM-inferred,
  verifier-certified contract suffices to check regression-equivalence of
C}, a language of large, long-lived systems where a
routine change can silently regress a caller.

Beyond staying safe, a change is a
regression only if something that uses the function can tell.
Figure~\ref{fig:motivating} shows the kind of edit we target.
\texttt{shrink} trims a box by a fixed margin; the old version forgets
to shrink the height, and the new one fixes it. The caller
\texttt{usable\_width} reads only the width, so it returns the same
value before and after the fix. \textit{Asking whether the two \texttt{shrink}
versions are equivalent is the wrong question}: the height really does
change, so a check over the whole of \texttt{*o} reports a difference,
flagging the fix itself. What matters is whether \texttt{shrink}
regresses on the width, the only part \texttt{usable\_width} observes,
and the contract $C$ in the figure captures exactly that. Enforced
against the old \texttt{shrink}, $C$ pins the width
(\texttt{ensures o->w == b->w - 2}) and leaves the height free. The
regression check replaces the old \texttt{shrink} by $C$ and runs the
new \texttt{shrink} concretely, comparing only the width: because $C$ is
precise exactly where the caller looks and silent elsewhere, the two
agree and the height change is invisible. A full specification would
have to pin the height too, the one part the fix changed and no caller
reads, making the proof obligation much harder for the contract author.

\begin{figure}[t]
  \centering
  {\footnotesize\textbf{Caller (unchanged in both versions)}}
\begin{lstlisting}[style=cdiff,numbers=left,numberstyle=\tiny\color{gray},numbersep=4pt]
int usable_width(const box_t *b) {
  box_t o;
  shrink(b, &o);
  return o.w;          // reads w, never h
}
\end{lstlisting}
  \begin{minipage}[t]{0.49\linewidth}
    \centering{\footnotesize\textbf{shrink (OLD)}}
\begin{lstlisting}[style=cdiff,numbers=left,numberstyle=\tiny\color{gray},numbersep=4pt,
  moredelim={[is][\color{red!70!black}\bfseries]{@}{@}}]
void shrink(const box_t *b,
            box_t *o) {
  o->w = b->w - 2;
  @o->h = b->h;@
}
\end{lstlisting}
  \end{minipage}\hfill
  \begin{minipage}[t]{0.49\linewidth}
    \centering{\footnotesize\textbf{shrink (NEW)}}
\begin{lstlisting}[style=cdiff,moredelim={[is][\color{green!50!black}\bfseries]{@}{@}}]
void shrink(const box_t *b,
            box_t *o) {
  o->w = b->w - 2;
  @o->h = b->h - 2;@
}
\end{lstlisting}
  \end{minipage}

  {\footnotesize\textbf{Inferred contract $C$ on \texttt{shrink}
  (enforced against OLD)}}
\begin{lstlisting}[style=cdiff]
requires \valid(b) && \valid(o);
assigns  o->w, o->h;
ensures  o->w == b->w - 2;     // nothing about o->h
\end{lstlisting}
  \caption{A property-relative regression check. The fix corrects
    \texttt{o->h}; the caller \texttt{usable\_width} reads only
    \texttt{o->w}, so it cannot observe the change. The partial contract
    $C$ pins \texttt{o->w} and stays silent on \texttt{o->h}: it is
    enforced against the old \texttt{shrink}
    ($\sem{\mathtt{shrink}_1}\subseteq\sem{C}$) and then replaces it in
    the regression check while the new \texttt{shrink} runs concretely,
    so the two agree on \texttt{o->w}. Comparing the whole of
  \texttt{*o} instead would flag the fix as a difference.}
  \label{fig:motivating}
\end{figure}

We make three contributions:

\begin{itemize}
  \item \textbf{C1 (Soundness).} A verifier-independent theorem that
    proves
    abstracting only the \emph{old} version by a single partial,
    one-sided contract suffices to prove \emph{property-relative
    regression-equivalence}, together with a depth-bounded lift for
    nested calls (\S\ref{sec:methodology}, Theorem~\ref{thm:t1}).

  \item \textbf{C2 (System).} An end-to-end pipeline of three
    tools: \ourtool{} itself scaffolds the regression harnesses,
    \textsc{scribe} runs the LLM-based propose-and-certify
    contract-inference loop,
    and \textsc{regver} discharges the regression-equivalence checking
    against the enforced contracts (\S\ref{sec:methodology}); to our
    knowledge the first to combine LLM-based contract inference with
    regression-equivalence verification.

  \item \textbf{C3 (Relative contract-tightness comparator).}
    \textsc{judge}, a semantic comparator that decides admit-set
    inclusion in both directions between two contracts and returns a
    four-way verdict (equivalent, one strictly tighter, or incomparable)
    with counterexample witnesses; on enforced contracts this yields a
    refinement ordering, so we can measure how tight an inferred
    contract is rather than only whether it verifies (\S\ref{sec:judge},
    Theorem~\ref{thm:t2}).

\end{itemize}

We evaluate C1, C2, and C3 through
two questions, stated and answered in \S\ref{sec:eval}: which
counterexample modality best drives frame-condition inference
(\emph{feedback modality}); and whether targeting a
partial, property-relative contract rather than a full specification
makes inference more tractable, and where its returns plateau
(\emph{tractability}).

The remainder of the paper is organized as follows.
\S\ref{sec:background} sets out the relational semantics, the three
equivalence notions, and function contracts with the frame rule;
\S\ref{sec:related} surveys related work; \S\ref{sec:methodology}
presents the soundness theorem (Theorem~\ref{thm:t1}) and its
depth-bounded lift, the system, and the contract-tightness procedure
(Theorem~\ref{thm:t2});
\S\ref{sec:eval} reports the experimental evaluation and analysis,
including threats to validity (\S\ref{sec:limitations}); and
\S\ref{sec:conclusion} concludes with open directions for future work.

\section{Preliminaries}\label{sec:background}

\subsection{Program semantics}\label{subsec:semantics}
\begin{definition}[Program semantics]\label{def:semantics}
  We model each version of a function as a relation
  $\sem{f}\subseteq X\times Y$ from input pre-states $X$ to output
  post-states $Y$, where $(x,o)\in\sem{f}$ means some terminating run of
  $f$ on $x$ yields $o$. The pre-state $X$ subsumes the responses of any
  external nondeterministic callee the two versions share,
  so a single $x$
  resolves those responses identically on both sides; only
  nondeterminism internal to the function itself leaves $\sem{f}(x)$
  multi-valued. We write
  $\sem{f}(x)\triangleq\{o\mid(x,o)\in\sem{f}\}$ for the set of outputs
  on $x$, which is empty when $f$ does not terminate on $x$.
\end{definition}
The spaces $X$ and $Y$ range over a common store, so a pre-state and a
post-state can be compared; this is what lets the frame clause below
speak of the locations a function leaves unchanged.

\subsection{Equivalence notions}\label{subsec:equiv}

Fix an old version $f_1$ and a new version $f_2$. The three notions
stated informally in \S\ref{sec:intro} order by how much of each
execution they compare~\cite{godlin2009rv,godlin2013rv}.

\begin{definition}[Equivalence notions]\label{def:equiv}
  \begin{itemize}
    \item \emph{Full equivalence}:
      $\sem{f_1}=\sem{f_2}$, so the versions terminate on the same inputs
      and produce the same outputs everywhere, old bugs included.
    \item \emph{Partial equivalence}
      relaxes termination: $\sem{f_1}(x) \neq \emptyset \Rightarrow
      \sem{f_1}(x)=\sem{f_2}(x)$, so only inputs on which $f_1$ terminates
      need to agree.
    \item \emph{Conditional (regression)
      equivalence} relaxes further to a
      property $\phi$ marking the behavior worth preserving (we write
      $\sem{f_1}(x) \models \phi$ when $f_1$'s runs on $x$ satisfy it),
      asking for agreement only where the old version was already correct:
      $\sem{f_1}(x) \models \phi \Rightarrow \sem{f_1}(x) = \sem{f_2}(x)$.
      Where $f_1$ violated $\phi$ the new version is free to differ, and
      that freedom is where the fix lives.
  \end{itemize}
\end{definition}

We target this conditional notion throughout. \S\ref{sec:methodology}
makes it precise on an explicit in-scope set and relaxes the output
equality to a property-relative obligation $E_\phi$, recovering the
definition above as the special case where $\phi$ observes the whole
output. We denote this relation we target \textit{safety-preserving
conditional equivalence}.

\subsection{Contracts and the Frame Rule}\label{subsec:contracts}
\begin{definition}[Function contract and admit set]\label{def:contract}
  A \emph{function contract}
  is a Hoare-style
  triple~\cite{hoare1969axiomatic,meyer1992dbc}: a precondition
  (\texttt{requires}), a postcondition (\texttt{ensures}), and a frame
  (\texttt{assigns}) naming the locations the function may modify. We write
  $C\triangleq(\text{Pre},\text{Post},\text{Assigns})$, with
  $\text{Pre}\subseteq X$ a set of input states and $\text{Post}\subseteq
  X\times Y$ relating inputs to outputs over the state spaces $X,Y$ defined
  above; throughout, we identify a set such as $\text{Pre}$ with its
  characteristic predicate, writing $\text{Pre}(x)$ for $x\in\text{Pre}$.
  Read in that relational model, the contract denotes
  its \emph{admit set}, the input/output pairs it permits,
  \[
    \begin{aligned}
      &\sem{C}\;\triangleq\;\{(x,o)\in X\times Y \mid \text{Pre}(x)\Rightarrow
        (\text{Post}(x,o)\ \wedge{}\\
      &\quad o\ \text{agrees with}\ x\ \text{outside}\ \text{Assigns})\},
    \end{aligned}
  \]
  with image $\sem{C}(x)$ the outputs it allows on $x$; outside its
  precondition the contract permits anything, so $\sem{C}(x)=Y$ when
  $\neg\text{Pre}(x)$.
\end{definition}
A modular verifier supports two operations on such a contract, used here
only through their effect on $\sem{C}$.
\emph{Enforcement} checks that the body of $f$ satisfies its
contract: every execution from a \texttt{requires}-state ends in an
\texttt{ensures}-state and writes only \texttt{assigns} locations. This
is exactly the inclusion $\sem{f}\subseteq\sem{C}$, so a passing
enforcement makes $C$ an over-approximation of $f$.
\emph{Replacement} abstracts each call to $f$ by the contract: the
call site asserts \texttt{requires}, havocs the locations named in
\texttt{assigns}, and assumes \texttt{ensures}, so the call ranges over
$\sem{C}(x)$ in place of $\sem{f}(x)$; since $\sem{f}\subseteq\sem{C}$
once enforced, a proof under replacement is sound for the real $f$.

Replacement is an
over-approximation: a proof obtained under it is sound for the real
function, but a reported failure may be \emph{spurious}, because
havocking the assigned locations can hand the caller an output the real
body never produces. The frame clause is what keeps replacement precise.
Locations not named in \texttt{assigns} are snapshotted before the body
and asserted unchanged after, so without an explicit \texttt{assigns} the
verifier must havoc every reachable location, and the abstraction becomes
too coarse to support a non-trivial equivalence proof. One caveat keeps
enforcement honest: if a contract's \texttt{requires} can never hold, the
body is never entered and the check passes \emph{vacuously}, testing
nothing. Our pipeline guards against this by reporting such a proof as
inconclusive rather than as a success, so the soundness of enforcement
never rests on an empty scope.

\section{Related Work}\label{sec:related}

\subsection{Regression Verification and Equivalence Checking}

Regression Verification Tool (RVT)~\cite{godlin2009rv,godlin2013rv}
pioneered modular regression verification, abstracting equivalent
matched callees with shared uninterpreted functions so the proof
holds for any consistent interpretation. Reve /
LLREVE~\cite{felsing2014reve,klebanov2018reve} encodes paired
programs as constrained Horn clauses and infers coupling predicates
relating the two versions' states with an external solver (Eldarica,
Z3), while SymDiff~\cite{lahiri2012symdiff} and
DAC~\cite{lahiri2013dac} relate the two versions by a relational
summary in the mutual-summary tradition~\cite{hawblitzel2013mutual}:
a single predicate over the disjoint union of both versions' inputs
and outputs that characterizes their behavioral relationship.
ARDiff~\cite{badihi2020ardiff} prunes common code by iterative
abstraction-refinement, but only for Java; PEQcheck and
PEQtest~\cite{jakobs2021peqcheck,jakobs2022peqtest} target localized
refactorings through segment-level encodings; and
PASDA~\cite{glock2024pasda} adds partition-based heuristics for the
undecided cases. Differential symbolic
execution~\cite{person2008dse,palikareva2016shadow,spatch2024,cozy2025}
characterizes program differences by exploring both versions'
symbolic state spaces, usually returning concrete divergence-inducing
inputs but no reusable artifact, as does LLM-driven test generation:
Mokav~\cite{etemadi2024mokav} pairs it with differential testing for
behavioral separation, and UnitTenX~\cite{charalambous2025unittenx}
drives agents over an ESBMC backend to build unit tests for legacy
code, both producing tests rather than a proof certificate.
Concurrent work by Sarker
et~al.~\cite{sarker2026quantitativesymbolicpatchimpact} takes
differential symbolic execution in an orthogonal direction,
\emph{quantifying} how much of the input domain a patch changes
rather than proving equivalence. Run as a soundness oracle on
EqBench-C our checker surfaced more mislabeled \texttt{Eq} pairs than
they report (nine versus five), while we target a sound modular
equivalence proof, not an impact measure.

Throughout this line the callee abstraction is supplied or derived from
both bodies in lockstep, presupposing a complete characterization of
what is abstracted. The closest exception is regression verification by
impact summaries~\cite{backes2013impact}, which restricts the proof to
change-relevant behavior: a control- and data-dependence analysis marks
the impacted statements and keeps only their path-condition constraints
from a whole-program symbolic execution. That summary is exact, so the
check is sound and complete to the depth bound, but for the same reason
the analysis must see all the code it reasons about. We instead
\emph{abstract}
the old version alone by a weak one-sided contract (one constraining the
old version, never the new), \textit{a reusable and enforceable artifact}
rather
than a path-condition partition tied to one diff, \textit{scoped to what the
property observes} rather than to what the change touches, and discharged
once bottom-up then reused at every call site and in every later check
(\S\ref{sec:methodology}).

\subsection{Contract, Invariant, and Frame Inference}

Inferring functional specifications with LLMs has advanced quickly. The
closest are ACSL-for-C tools validated by Frama-C: AutoSpec's
accept-or-reject WP feedback~\cite{wen2024autospec}, ``Specify What?''
priming on EVA and Pathcrawler data~\cite{granberry2024specifywhat}, and
Preguss's iterative runtime-error-alarm refinement~\cite{wang2025preguss},
the nearest analogue to our BMC-trace-driven loop on a Frama-C backend.
Other LLM-based generate+validate
loops~\cite{tihanyi2023esbmcai,richter2025beyondpost,pirzada2026converusingcontractsloop}.
Non-LLM inference~\cite{amilon2025autodeduct,amilon2024mindreading} and
symbolic frame
inference~\cite{rakamaric2008frames,calcagno2009biabduction,quiver2024}
predate these; we contribute the BMC-backed counterpart, with concrete-trace
feedback drawn from ESBMC's frame-rule pass.

Two gaps separate this work from all of the above. First, every such
tool targets a single version; none reuses the inferred contract as a
callee abstraction in a two-version regression-equivalence proof, the
gap we close. Second, they judge a specification by whether it
verifies, not how tight it
is~\cite{granberry2024specifywhat,beg2026acsleval}; to our knowledge
none defines a semantic notion of contract tightness or a decision
procedure for contract subsumption, which \S\ref{sec:judge}
introduces. The CE-guided LLM loop is by now established for
invariants~\cite{lam4inv2024,clause2inv2025,bali2025,su2026cill,dafnypro2026,dafnyannotator2024,daisy2025,wei2025quokka}
and descends from counterexample-guided inductive
synthesis~\cite{sketch2008,garg2016icedt,padhi2016pie,si2018code2inv},
but those target invariants, not the function-level frame conditions we infer.

\section{Methodology}\label{sec:methodology}

\subsection{One-Sided Contracts for Regression Equivalence (C1)}

Recall from Definition~\ref{def:semantics} that each version is a
relation $\sem{f}\subseteq X\times Y$ of terminating runs, with image
$\sem{f}(x)$, and from Definition~\ref{def:contract} that a contract
denotes a relation $\sem{C}\subseteq X\times Y$, with image $\sem{C}(x)$
the outputs it allows on $x$. The abstraction $\alpha$ of
\S\ref{sec:intro}
is just a contract: taking $\sem{\alpha}=\sem{C}$, its
over-approximation requirement $\sem{f_1}\subseteq\sem{\alpha}$ becomes
$\sem{f_1}\subseteq\sem{C}$.

\begin{definition}[Scope and equivalence obligation]\label{def:scope}
  A property $\phi$ fixes a \emph{scope}
  $\mathit{Sc}\triangleq\{x\in X\mid\sem{f_1}(x)\models\phi\}$ of
  in-scope inputs (the $x$ on which every run of $f_1$ stays within
  $\phi$, e.g.\ runs safely), and an \emph{equivalence obligation}
  $E_\phi(o_1,o_2)$ that holds when $o_1$ and $o_2$ agree on the part of
  the output $\phi$ observes (for example, the return value and the
  locations a caller reads).
\end{definition}
We work under partial correctness and three standing assumptions, made
explicit because the theorem rests on them: (A1) only terminating runs
populate $\sem{\cdot}$, so the claim is over inputs on which both
versions terminate~\cite{godlin2013rv} (in particular, a change that
  makes $f_2$ diverge where $f_1$ terminated is outside the guarantee, and
detecting such termination regressions is left to future work); (A2) the
frame rule is
sound (a body that passes enforcement writes no location outside
$\text{Assigns}$, even under aliasing); and
(A3) the verifier that
discharges enforcement and replacement is sound, so every verdict it
certifies holds in the relational model (Definition~\ref{def:contract}).
External nondeterminism
needs no separate assumption: a shared callee's responses are part of
the input $x$ (Definition~\ref{def:semantics}), so they are resolved
identically on both sides by construction and cannot themselves induce a
divergence.

\begin{definition}[Contract preserving
  regression-equivalence]\label{def:preserves}
  Given two versions $f_1, f_2$ of a function and a property $\phi$, a
  contract $C$ \emph{preserves regression-equivalence with respect to
  $\phi$} iff:
  \begin{enumerate}
    \item \textbf{Enforcement.} $\sem{f_1}\subseteq\sem{C}$: $C$
      over-approximates $f_1$, so every run of $f_1$ is one $C$ permits.
    \item \textbf{Sufficiency for $\phi$.} For every $x\in\mathit{Sc}$,
      every $o\in\sem{C}(x)$, and every $o_2\in\sem{f_2}(x)$, we have
      $E_\phi(o,o_2)$: every output the replacement of $f_1$ by $C$ may
      produce meets the regression obligation against $f_2$.
  \end{enumerate}
\end{definition}

Condition~2 mirrors replacement: the abstracted call ranges over exactly
$\sem{C}(x)$, the outputs $C$ allows on $x$. We need not separately
assume $\mathit{Sc}\subseteq\text{Pre}$: dropping it only strengthens
condition~2's hypothesis (outside $\text{Pre}$, $\sem{C}(x)=Y$, so the
obligation ranges over more outputs, never fewer), and the soundness
proof never relies on $\mathit{Sc}\subseteq\text{Pre}$.
\begin{definition}[$\phi$-regression-equivalence]\label{def:regequiv}
  $f_2$ is \emph{$\phi$-regression-equivalent} to $f_1$ iff for every
  $x\in\mathit{Sc}$, every $o_1\in\sem{f_1}(x)$, and every
  $o_2\in\sem{f_2}(x)$ the obligation $E_\phi(o_1,o_2)$ holds.
\end{definition}
For
deterministic functions (each $\sem{f}(x)$ a single output, shared
  external nondeterminism already folded into $x$ by
Definition~\ref{def:semantics}), this
coincides with the conditional equivalence
$\sem{f_1}(x)\models\phi\Rightarrow\sem{f_1}(x)=\sem{f_2}(x)$ of
Definition~\ref{def:equiv}, equality relaxed to $E_\phi$; for genuinely
multi-valued $\sem{f}(x)$
the pairwise form is strictly stronger.


Intuitively, the contract only needs to over-approximate the old
version, while being precise on the behavior observed by callers.
This allows us to replace the old implementation during verification
without introducing unsoundness.

\begin{theorem}[Soundness of one-sided contract abstraction]\label{thm:t1}
  Let $f_1,f_2$ be two versions of a function, $C$ a contract, and $\phi$
  a property with scope $\mathit{Sc}$ and obligation $E_\phi$. If $C$
  preserves regression-equivalence with respect to $\phi$
  (Definition~\ref{def:preserves}), then $f_2$ is
  $\phi$-regression-equivalent to $f_1$ (Definition~\ref{def:regequiv}).
\end{theorem}
\begin{proof}
  Fix $x\in\mathit{Sc}$, $o_1\in\sem{f_1}(x)$, and
  $o_2\in\sem{f_2}(x)$. By condition~1, $\sem{f_1}\subseteq\sem{C}$, so
  $\sem{f_1}(x)\subseteq\sem{C}(x)$ and $o_1\in\sem{C}(x)$. Instantiate
  condition~2 with this $x$, the witness $o\mathbin{:=}o_1$, and the
  given $o_2$: its premises $o_1\in\sem{C}(x)$ and $o_2\in\sem{f_2}(x)$
  both hold, so $E_\phi(o_1,o_2)$. As $x,o_1,o_2$ were arbitrary, $f_2$
  is $\phi$-regression-equivalent to $f_1$.
\end{proof}

The proof turns on one step: enforcement
(condition~1) places the real
output $o_1$ inside $\sem{C}(x)$, and replacement (condition~2)
quantifies the obligation over all of $\sem{C}(x)$, so it specializes
to $o_1$. Condition~1 is what makes the proof sound; condition~2 is the
only place $C$ must be tight, and everywhere else it may stay weak.
The relational-summary abstractions of
SymDiff~\cite{lahiri2012symdiff} and DAC~\cite{lahiri2013dac} and the
shared-uninterpreted-function abstraction of RVT~\cite{godlin2013rv} do
not allow this freedom, since they need a complete characterization of
the procedure. The remaining obligations are discharged by the standing
assumptions (pointer aliasing across $\text{Assigns}$ and
non-termination) and by the semantics (external nondeterminism is folded
into the shared input, Definition~\ref{def:semantics}); these are
standard for Hoare reasoning with a frame rule. \textit{To our knowledge this is
  the first such formalization for a one-sided contract inside a
  regression query, and it is what makes the contract small enough for an
LLM to infer}.

How our pipeline discharges conditions~1--2 as two concrete ESBMC queries
is described in \S\ref{subsec:system}.

Replacement is an over-approximation, so the converse of
Theorem~\ref{thm:t1} can fail: the equivalence check may report a
\emph{spurious} divergence, a witness $o\in\sem{C}(x)$ that $C$ permits
but $f_1$ never produces and on which $E_\phi$ fails. This is exactly
why condition~2 requires $C$ to be precise on the
$\phi$-observed output: looseness there admits an observably-disagreeing
$o$ and condition~2 fails even when $f_1$ and $f_2$ genuinely agree
(lost completeness, not soundness), whereas looseness off the observed
projection only enlarges $\sem{C}$ on locations $E_\phi$ ignores and
is harmless. That harmless region is the partial-spec freedom, and it
keeps the LLM step practical.

\begin{remark}[Non-vacuity is machine-discharged]\label{rem:vacuity}
  Theorem~\ref{thm:t1} is content-free if $\mathit{Sc}=\varnothing$:
  an unsatisfiable $\text{Pre}$ makes condition~1 hold vacuously and
  the conclusion vacuously true. Our approach does not rely on
  the user to rule this out. Instead it employs a vacuity probe to ensure
  that the contract conditions are not vacuously
  true~\cite{beer1997efficient}.
\end{remark}

\subsection{Depth-Bounded Contract Closure}\label{sec:depth-bounded-closure}

Theorem~\ref{thm:t1} governs the boundary between one abstracted procedure and
its caller.
Regression targets in legacy code transitively touch
hundreds of helpers, and inferring partial contracts for all of them
in a single LLM pass is intractable. We extend T1 by bounding the
contract closure to a configurable depth $N$ from the seed target
$g_0$. Let $d(g)$ be the depth of $g$, its distance from the seed $g_0$.
Functions at distance $d(g) < N$ along the static call graph
receive precise partial contracts of T1's shape; functions at $d(g) =
N$ receive deliberately weak \emph{boundary contracts}
$C^{\partial}_g$ (precondition \texttt{true}, postcondition \texttt{true},
and a frame that havocs every location reachable from the call);
functions at $d(g) > N$
receive no contract and are not analyzed in this pass. Under
contract replacement, boundary contracts
substitute at depth-$N$ call sites and deeper bodies are never
expanded, so the cost of verifying $g_0$ is bounded by the depth-$N$
closure independently of the full callee depth.

Write $\sem{g\mid\sigma}$ for the denotation of $g$ under a \emph{callee
environment} $\sigma$ that replaces each callee $h$ of $g$ by a relation.
Here $(\,\cdot\,)_h$ is the family indexed by the callees $h$ of $g$, so
$\sem{g\mid(\sem{h})_h}=\sem{g}$ sends each callee to its own denotation and
$\sem{g\mid(\sem{C_h})_h}$ replaces every callee by its contract. The
semantics is \emph{monotone} in $\sigma$ (a callee
relation occurs only positively): $\sigma\subseteq\sigma'$ pointwise implies
$\sem{g\mid\sigma}\subseteq\sem{g\mid\sigma'}$.

\begin{lemma}[Depth-$N$ closure soundness]\label{lem:depth}
  With contracts assigned as above, suppose
  \begin{enumerate}
    \item[(i)] \emph{(boundary)} each $C^{\partial}_g$ at $d(g)=N$ havocs a
      superset of $g$'s footprint, so
      $\sem{g}\subseteq\sem{C^{\partial}_g}$; and
    \item[(ii)] \emph{(interior)} each $g$ at $d(g)<N$ passes enforcement
      against its callees' contracts,
      $\sem{g\mid(\sem{C_h})_h}\subseteq\sem{C_g}$.
  \end{enumerate}
  Then $\sem{g_0}\subseteq\sem{C_{g_0}}$: the seed contract
  over-approximates the
  seed without expanding any body below depth $N$.
\end{lemma}
\begin{proof}
  Induct on the call relation of the depth-$\le N$ closure, each callee before
  its caller. \emph{Base.} A boundary node ($d(g)=N$) is covered by~(i),
  $\sem{g}\subseteq\sem{C^{\partial}_g}$, with no enforcement query.
  \emph{Step.} For an interior $g$, every callee $h$ has $d(h)\le d(g)+1\le N$,
  so $h$ lies in the closure with a contract $C_h$; the induction hypothesis
  gives $\sem{h}\subseteq\sem{C_h}$, monotonicity gives
  $\sem{g}=\sem{g\mid(\sem{h})_h}\subseteq\sem{g\mid(\sem{C_h})_h}$, and~(ii)
  gives $\sem{g}\subseteq\sem{C_g}$. The order is well-founded when the closure
  is acyclic; an interior recursive $g$ replaces its self-call by $C_g$ in~(ii),
  sound by the standard partial-correctness contract rule (A1 makes each
  terminating run's recursion finite). Take $g=g_0$.
\end{proof}

Lemma~\ref{lem:depth} supplies condition~1 of T1 for $g_0$ at the cost of the
depth-$N$ closure alone; condition~2 is discharged by the seed's equivalence
harness (\S\ref{subsec:system}), independent of $N$. \textit{Boundary
  contracts need no
  enforcement step, so deepening the boundary trades precision, not
  soundness, and
T1's partial-spec relaxation carries over to each interior $C_g$}. Iterative
deepening re-seeds a boundary function with bound $N$, turning its case~(i) into
case~(ii) one stratum deeper; the single knob $N$ trades inferred-contract
precision against verification scalability.

\subsection{System Architecture (C2)}\label{subsec:system}

Our tools realize the method as a three-stage pipeline.
\ourtool{} parses the C project with tree-sitter, identifies the
target functions, and scaffolds their harnesses. \textsc{scribe} then
runs the LLM contract-inference loop over the interior functions through
a JSON-schema interface, closing the depth-$N$ frontier with the
mechanical havoc boundary contracts of \S\ref{sec:depth-bounded-closure}.
Finally, \textsc{regver} stages the old version, either from a git commit or a
directory snapshot, and renames its symbols on a temporary copy so that
both versions can coexist in a single harness.

From the enforced contracts, \textsc{regver} compiles condition~2 of
Definition~\ref{def:preserves} into one equivalence harness. On a shared
nondeterministic input $x\in\mathit{Sc}$ it replaces the old version by
its contract: the call site asserts \texttt{requires}, havocs
\texttt{assigns}, and assumes \texttt{ensures}, so the result ranges over
$\sem{C}(x)$. It then runs the new version $f_2$ concretely and asserts
$E_\phi$ between the two outputs. That single verification condition is
exactly condition~2, while enforcement (condition~1) is a separate query
placing every real $f_1$-output inside $\sem{C}(x)$. The check thus
reduces to two ESBMC queries, and $f_1$'s body is never expanded inside
the equivalence proof.

Theorem~\ref{thm:t1} is stated over an abstract contract checker, and
we instantiate its enforcement and replacement by ESBMC's enforce and
replace modes, whose documented replace-mode behavior supplies the
over-approximation the proof relies
on~\cite{pirzada2026converusingcontractsloop}. The theorem leaves the
post-state property $\phi$ abstract, and in the implementation we fix
it to the safety properties ESBMC encodes on every
harness---\textbf{memory safety (spatial bounds and pointer- and
  dereference-validity, covering use-after-free), division by zero,
  arithmetic over- and underflow, memory leaks, and undefined-behavior
shifts}.
The same instantiation inherits the non-vacuity guard of
Remark~\ref{rem:vacuity} from ESBMC's vacuity probe.

Inference is driven by the very counterexamples the verifier already
produces. When ESBMC enforces a proposed contract $C^{(0)}$ against
$f_1$'s body and the frame check fails, its frame-rule pass recovers
the offending write (file, line, and address expression) from the
counterexample trace, and \textsc{scribe} feeds that trace, not a
bare verification-condition failure, into the next prompt.

\subsection{Semantic Contract-Tightness Comparison (C3)}\label{sec:judge}

The published LLM contract-inference baselines
(AutoSpec~\cite{wen2024autospec}, ``Specify
  What?''~\cite{granberry2024specifywhat},
Preguss~\cite{wang2025preguss}) do not measure contract tightness
semantically: \textit{they report verification success rate, or at
most count annotations by type}. Any such syntactic proxy is
inflatable: \textit{a trivially true postcondition adds to the count
without adding constraint}. We replace it with a decision procedure
for contract subsumption.

For contracts sound $A, B$ on a common function over the state spaces $X, Y$
of Definition~\ref{def:semantics}, tightness is set inclusion on their admit
sets: $A \le_{\text{tight}} B \iff \sem{A} \subseteq \sem{B}$. \textsc{judge}
(C3) decides this ordering, splitting $\sem{C}$ into its postcondition
and frame parts. The postcondition part is two validity checks:
\begin{itemize}
  \item \textbf{Check~1:}\;\; $\forall x, o.\; \text{admits}_A(x, o)
    \Rightarrow \text{admits}_B(x, o)$ \quad (is $\sem{A} \subseteq \sem{B}$?)
  \item \textbf{Check~2:}\;\; $\forall x, o.\; \text{admits}_B(x, o)
    \Rightarrow \text{admits}_A(x, o)$ \quad (is $\sem{B} \subseteq \sem{A}$?)
\end{itemize}
where $\text{admits}_C(x, o) \equiv \neg\,\text{Pre}_C(x) \lor
\text{Post}_C(x, o)$ is membership in $\sem{C}$ with the frame conjunct
dropped (a superset test); frame containment on $\text{Assigns}$ is
checked separately (syntactically in v1, see
\S\ref{sec:limitations}). When $A$ and $B$
share their precondition and frame (the warm-start,
partial-versus-strengthened case Theorem~\ref{thm:t2} ranges over), a
passing Check~1/Check~2 establishes $\sem{A}\subseteq\sem{B}$. These
checks are sound in that direction but not complete: because
$\text{admits}_C$ drops the frame, a \emph{failing} check can be
spurious (its witness may lie outside the shared precondition/frame
region), so judge may under-report tightness, never over-report it.
Under the same shared-precondition/frame assumption, the four outcomes
of (Check~1, Check~2) classify the pair: $(\checkmark, \checkmark)$
is \emph{equivalent}
($\sem{A} = \sem{B}$); $(\checkmark, \times)$ is \emph{$A$ strictly
tighter}; $(\times, \checkmark)$ is \emph{$B$ strictly tighter};
$(\times, \times)$ is \emph{incomparable}. Each failing check yields a
counterexample witness: a concrete $(x, o)$ admitted by one contract
and rejected by the other. The incomparable case ships two witnesses,
one per direction, surfacing the precise behavior each contract
admits that the other does not. Each check is a pure validity query
over a nondeterministic input and output, with no body of $f$
involved (neither enforcement nor replacement); our implementation
discharges each as a single SMT query.

\textbf{Relative ordering, not absolute completeness.} \textsc{judge}
measures the order between $\sem{A}$ and $\sem{B}$, not how close either
is to $\sem{f}$. A contract it
finds strictly tighter may still loosely over-approximate $f$, or not be
valid for $f$ at all; validity is settled separately by enforcement
(Theorem~\ref{thm:t1}, condition~1), and we always report the two
together. \textit{A verdict between two enforced contracts is thus a statement
about sound contracts}, the tighter side admitting strictly fewer
behaviors while still covering all of $f$'s; between unenforced
contracts it is only a relative-over-approximation result, labelled as
such in the evaluation tables.

\textbf{Tightness is the classical refinement
ordering~\cite{liskov1994subtyping,back1998refinement}.} When both
contracts \emph{are} enforcement-checked, the ordering says more than
set inclusion: the tighter contract proves at least as much,
discharging every obligation the looser one does. We state this as an
extension of Theorem~\ref{thm:t1}, reusing its scope $\mathit{Sc}$,
obligation $E_\phi$, and condition~2 (sufficiency). We keep
Theorem~\ref{thm:t1}'s scope $\mathit{Sc}$ unchanged: the hypothesis
$\sem{A}\subseteq\sem{B}$ already gives $\sem{A}(x)\subseteq\sem{B}(x)$
at every $x$, so no narrowing to the preconditions is needed (in the
  warm-start, partial-versus-strengthened comparison $A$ and $B$ share
$\text{Pre}$ anyway).

\begin{theorem}[Tightness refines proof power]\label{thm:t2}
  Let $A, B$ be contracts for $f_1$ that both pass enforcement
  ($\sem{f_1}\subseteq\sem{A}$ and $\sem{f_1}\subseteq\sem{B}$), with
  $\sem{A}\subseteq\sem{B}$ (judge's Check~1: $A$ is at least as tight as
  $B$). Then every regression-equivalence query against $f_2$ that
  $\text{replace}(B)$ discharges, $\text{replace}(A)$ discharges
  too (condition~2 of Theorem~\ref{thm:t1} for $B$ implies condition~2
  for $A$), and the resulting proof is sound for the real $f_1$ and
  $f_2$.
\end{theorem}
\begin{proof}
  Fix $x\in\mathit{Sc}$. Since $\sem{A}\subseteq\sem{B}$, also
  $\sem{A}(x)\subseteq\sem{B}(x)$. Condition~2 for $B$ requires
  $E_\phi(o,o_2)$ for every $o\in\sem{B}(x)$ and $o_2\in\sem{f_2}(x)$;
  restricting to $o\in\sem{A}(x)$ gives condition~2 for $A$. As $A$
  passes enforcement, Theorem~\ref{thm:t1} turns condition~2 for $A$
  into the regression-equivalence of $f_1$ and $f_2$.
\end{proof}

The converse fails, and that is the point: a looser contract proves
no more, and may admit an output $f_1$ never produces that raises a
spurious divergence the tighter one rules out. So for
enforcement-sound contracts judge's verdict is a refinement ordering:
the strictly tighter side discharges at least as many regression
queries, and strictly more whenever it eliminates such a spurious
counterexample.

\section{Evaluation}\label{sec:eval}

We evaluate the contributions of \S\ref{sec:methodology} around two
questions:

\begin{tcolorbox}[mybox,title=RQ1 - Feedback Modality]
  Which counterexample modality most effectively drives frame-condition
  inference: concrete BMC traces, verification-condition failures
  (AutoSpec), or runtime-error traces (Preguss)?
\end{tcolorbox}

\begin{tcolorbox}[mybox,title=RQ2 - Tractability]
  Given that the \textsc{scribe} loop produces a partial contract
  \textbf{(RQ1)}, how much
  additional tightness does continued strengthening recover beyond it,
  and do those returns plateau? This is the practical test of the
  partial-spec freedom of Theorem~\ref{thm:t1}: not whether a partial
  contract is cheaper to infer, but whether it is a good place to stop.
\end{tcolorbox}

The order is deliberate. RQ1 establishes that the inference loop
works at all, and RQ2
then asks whether the partial contract it produces is already tight
enough; RQ2 only makes sense once RQ1 is settled. We describe the
datasets, baselines, and metrics, then report results.

\subsection{Datasets}

\paragraph{EqBench-C.} The C subset of
EqBench~\cite{badihi2021eqbench}: 147 equivalent and 125
non-equivalent program pairs. EqBench is the most comprehensive public
equivalence-checking dataset, built to cover the hard constructs
earlier benchmarks omit (non-linear arithmetic, loops, floating point,
and string and array manipulation), which makes its C subset a
demanding soundness stress-test even though formal C tools have rarely
reported on it. Its pairs are labelled by input-output equivalence of
terminating runs, the same notion \textsc{regver}'s contract-free
harness decides, so the labels are valid ground truth for our
regression check. We use EqBench-C as a soundness baseline
for \textsc{regver}, not a cross-tool rate ranking: the question we
ask is the safe one, whether \textsc{regver} ever certifies a
non-equivalent pair as equivalent. We evaluate on EqBench-C rather
than the more recent EquiBench~\cite{wei2025equibench} because
EquiBench labels pairs by \emph{full input-output equivalence}
(Definition~\ref{def:equiv}), a relation finer than the
safety-preserving regression equivalence our methodology targets.

\paragraph{Frama-C-Problems.}
Frama-C-Problems~\cite{patnaik2020framacproblems} is a community set
of 51 small C programs ($\sim$20 LoC average), each carrying an
unproven guard assertion and 1--3 ACSL specifications.
AutoSpec~\cite{wen2024autospec} introduced it as a contract-inference
benchmark and Preguss~\cite{wang2025preguss} reused it, so it is the
standing benchmark for this task and both report per-program success
rates on it. We evaluate \textsc{scribe} on all 51 programs.

\paragraph{X509-parser.} X509-parser is a formally verified RTE-free X.509
parser ($>$1000 LoC, manually verified in ACSL by its authors over
five months), against which AutoSpec~\cite{wen2024autospec} reports
synthesized specifications for six functions, selected to span loops
with buffer arithmetic, switch-case control flow, inter-procedural
composition, and mixed-width integer and shift operations. We
evaluate \textsc{scribe} on those same six functions; the
depth-bounded closure of \S\ref{sec:depth-bounded-closure}
(Lemma~\ref{lem:depth}), a standing
part of our C2 contribution rather than an ablation knob, handles
their inter-procedural composition.

\subsection{Baselines}\label{subsec:baselines}

For frame inference (RQ1) we compare against
AutoSpec~\cite{wen2024autospec} (LLM+WP, VC-level feedback),
``Specify What?''~\cite{granberry2024specifywhat}
(LLM+WP+EVA/Pathcrawler prompt augmentation),
Preguss~\cite{wang2025preguss} (LLM+WP+RTE-trace iterative
refinement), and AutoDeduct~\cite{amilon2025autodeduct} (CHC+Eva, no
LLM). We compare against the numbers each tool reports in its
published evaluation, on the inputs it reports them for:
Frama-C-Problems for AutoSpec and Preguss, and the SV-COMP-derived
subset and X509-parser for AutoSpec. We run \textsc{scribe} on those
same inputs and report its success rate and iteration count
alongside theirs. For regression
verification we use EqBench-C only as a lightweight soundness check
on \textsc{regver}, without a numeric cross-tool ranking: the
established equivalence checkers (RVT~\cite{godlin2013rv},
  Reve~\cite{klebanov2018reve}, ARDiff~\cite{badihi2020ardiff},
SymDiff~\cite{lahiri2012symdiff}) decide partial rather than
safety-preserving conditional equivalence, and the ones with EqBench
results (ARDiff and its symbolic-execution successors) report on the
Java subset, so no like-for-like EqBench-C number exists to compare against.
We run \textsc{regver} on this set in two arms, a contract-free baseline and a
treatment arm in which \textsc{scribe} first infers a one-sided contract for
the changed function, and combine them as a sequential portfolio: the baseline
runs first and the treatment arm is invoked only on the pairs it leaves
undecided, the timeout and tool-error tail. Because the escalation gate is the
baseline's own abstention, observable at run time and independent of the
dataset label, the combination is a portfolio and not an oracle, and it
inherits every false positive the contract layer raises on that tail.

\subsection{Experimental Setup}

We evaluate with three LLMs chosen
for diversity: Opus 4.8~\cite{anthropic_claude_opus_4_8}
(closed-source frontier), Kimi K2.6~\cite{moonshotai_kimi_k2_6}
(open-weight frontier), and Qwen3.6-27B~\cite{qwen3.6-27b}
(open-weight, runnable on inexpensive hardware), all at temperature $0$.

For the formal verifier we use ESBMC
with Z3 (due to quantifier support). ESBMC runs under
$k$-induction with no fixed unwinding bound: inductively closed loops
are proved for every iteration and the rest unrolled in full,
completing on finite loops and otherwise timing out. Three limits are
fixed across every comparison run: a per-query ESBMC timeout and a
per-call LLM timeout, both 300s, and at most 5 refinement iterations
per target. ESBMC also runs under a 104GB per-process memory budget,
applied as a Linux control-group limit with swap off, so a run that
exceeds it stops at its memory ceiling rather than growing into swap;
a stopped run is a non-result, never a pass or a failure. A
non-result is a run that returns no verdict: a query timeout, a
memory-ceiling stop, an ESBMC crash or solver error, an unreadable
exit, or an LLM API failure. We report each success rate two ways:
over the cases ESBMC decided, and with every non-result charged as a
failure (to keep the lower bound conservative).

\subsection{Metrics and Mapping to RQs}

Each RQ pairs a metric with the dataset it is measured on.
\begin{itemize}
  \item \textbf{RQ1} (\textcolor{blue}{feedback modality}):
    contract-inference success rate and iterations to convergence on
    Frama-C-Problems and X509-parser, read against AutoSpec's and
    Preguss's published numbers.
  \item \textbf{RQ2} (\textcolor{blue}{tractability}): the round-by-round
    \textsc{judge} tightness verdict (the saturation curve) and pass rate
    of \textsc{scribe}'s partial contract against the warm-start
    \textsc{hardscribe} arm, on Frama-C-Problems.
\end{itemize}

\subsection{Results}

We report RQ1 and RQ2 from the three-model sweep on the inputs the
baselines publish, plus the EqBench-C soundness baseline.

\subsubsection{RQ1 (feedback modality)}

We ran \textsc{scribe} in its
default partial-spec mode on the two inputs the closest baselines report on,
the 51 Frama-C-Problems programs and the X509 parser, across all three
models. Table~\ref{tab:rq1-crosstool} gives the headline. On
Frama-C-Problems, \textsc{scribe} infers a verifier-discharged contract for
most programs on all three models, with the decided and conservative
success rates shown.
Non-results are few and cluster on the hard array and loop cases
(20 across the three models, mostly query timeouts, the rest ESBMC crashes,
unreadable exits, and one LLM API error; no run reached the 104GB ceiling).

Measured the way the baselines count, with every non-result a failure, these
conservative rates clear AutoSpec on every model and reach Preguss's
per-trial rate at the top of the range, from a single deterministic run at
temperature $0$ where Preguss averages three trials. The three tools use
different success oracles: \textit{AutoSpec targets functional
  correctness, Preguss
  freedom from runtime errors, and \textsc{scribe} a contract that passes both
enforcement and replacement}, which the next paragraph shows subsumes
Preguss's condition.

One feature makes the comparison fairer than the differing oracles suggest:
\textsc{scribe}'s enforcement turns on ESBMC's full safety set, covering the
runtime-error class Preguss proves absent, while also checking a functional
postcondition and caller-sufficiency, so a passing \textsc{scribe} result
certifies at least as much as a passing Preguss one.

Table~\ref{tab:rq1-refine} reports the per-step pass rates, with the same
pattern across the three models: enforcement is the binding gate and
replacement nearly always follows once it passes.
postcondition toward one the body satisfies.
Most targets converge in a
single iteration (Figure~\ref{fig:rq1-conv}). On the
X509 parser, \textsc{scribe} verifies 5 of 6 target functions with Opus and
Kimi and 3 of 6 with Qwen, against the 6 of 6 AutoSpec reports; the function
set is similar but not identical, so we read this as close rather than
head-to-head.

\begin{tcolorbox}[mybox,title=RQ1 Answer]
  Model-checker counterexamples alone provide an effective refinement signal
  for contract inference. Across the evaluated benchmarks \textsc{scribe}
  reaches verification rates comparable to prior approaches that additionally
  rely on weakest-precondition reasoning, value analysis, or runtime-error
  annotations, while most contracts converge after a single refinement
  iteration.
\end{tcolorbox}

\begin{table}[t]
  \centering
  \caption{RQ1 contract-inference success}
  \label{tab:rq1-crosstool}
  \begin{threeparttable}
    \begin{tabular}{lcc}
      \toprule
      Tool (oracle) & Frama-C-Problems & X509 \\
      \midrule
      AutoSpec (func.) & 31/51 (60.8\%)            & 6/6 \\
      Preguss (RTE)    & 122/153\tnote{a}\ (79.7\%) & n/a \\
      \midrule
      \textsc{scribe}, Opus 4.8        & 36/51 (70.6\% / 85.7\%)  & 5/6 \\
      \textsc{scribe}, Kimi K2.6       & 40/51 (78.4\% / 85.1\%)  & 5/6 \\
      \textsc{scribe}, Qwen3.6-27B     & 36/51 (70.6\% / 81.8\%)  & 3/6 \\
      \bottomrule
    \end{tabular}
    \begin{tablenotes}
      \footnotesize
    \item[a] Per-trial rate over three trials.
    \end{tablenotes}
  \end{threeparttable}
\end{table}

\begin{table}[t]
  \centering
  \caption{RQ1 per-step pass rates over refinement attempts, per model.}
  \label{tab:rq1-refine}
  \begin{tabular}{lcc}
    \toprule
    Model & Enforce & Replace \\
    \midrule
    Kimi K2.6   & 0.82 & 0.95 \\
    Opus 4.8    & 0.75 & 0.98 \\
    Qwen3.6-27B & 0.75 & 0.93 \\
    \bottomrule
  \end{tabular}
\end{table}

\begin{figure}[t]
  \centering
  \includegraphics[width=\columnwidth]{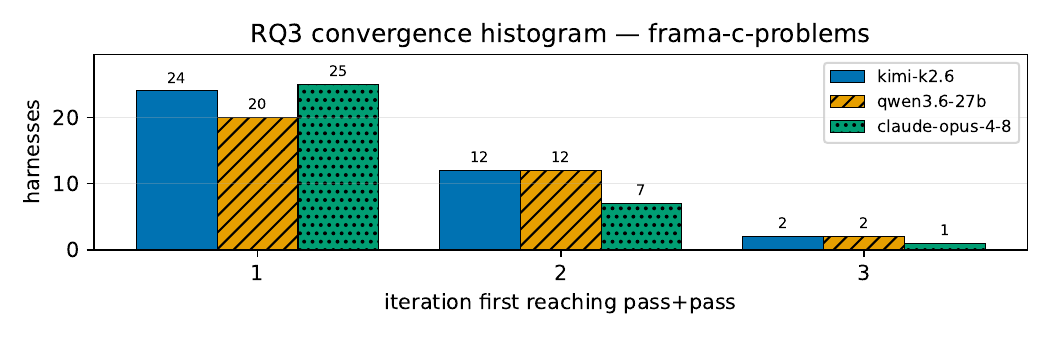}
  \caption{RQ1 convergence histogram per model on Frama-C-Problems: the
    iteration at which each harness first reaches a passing enforce and
  replace. Most targets converge in a single iteration.}
  \label{fig:rq1-conv}
\end{figure}

\subsubsection{RQ2 (tractability)}

RQ1 leaves the loop reliably producing a verifiable partial contract. We
seed the warm-start arm (\textsc{hardscribe}) from each model's partial
contract and apply up to five enforce-gated strengthening rounds, keeping
every accepted round so \textsc{judge} can measure how much tighter each
round is than the seed. Strengthening saturates fast in every model
(Figure~\ref{fig:rq2-sat}): for almost all orderable targets the contract is
already at its tightest at the seed (Table~\ref{tab:rq2}), and where
strengthening helps it does so in the first round or two and then stops, its
round-over-round gain reaching zero by round two or three.

The endpoint comparison agrees (Table~\ref{tab:rq2}): judging each model's
full contract against its own partial seed over the orderable pairs, the two
are equivalent for the large majority of targets and only a few are strictly
tighter after strengthening. The warm-start arm costs almost nothing in
soundness, keeping the last contract that verifies, so its pass rate matches
the partial arm to within one program in every model. The answer to RQ2 is
that caller-sufficiency already captures most of the attainable tightness: a
cheap partial contract is within a few targets of the strongest contract the
loop can reach, a direct validation of the partial-spec stopping point.

Two bounds keep this honest. About 40\% of round-level comparisons involve a
quantified or pointer-dereferencing postcondition \textsc{judge}'s
signature-only harness cannot order; we hold those out and fix the denominator
to the targets orderable in every round, so the saturation we report speaks
only for the contracts \textsc{judge} can compare across all rounds. And
the warm-start arm strengthens within the seed's clause shape, so this rapid
saturation shows little tightness is left within that shape, not that no tighter
contract of a different shape exists, which we leave to the from-scratch
study in \S\ref{sec:limitations}.

\begin{tcolorbox}[mybox,title=RQ2 Answer]
  Continued strengthening recovers little beyond the partial contract:
  caller-sufficiency already captures most of the attainable tightness, and
  the few targets that tighten do so within a round or two before the returns
  plateau, so the partial contract is a sound place to stop.
\end{tcolorbox}

\begin{figure}[t]
  \centering
  \includegraphics[width=\columnwidth]{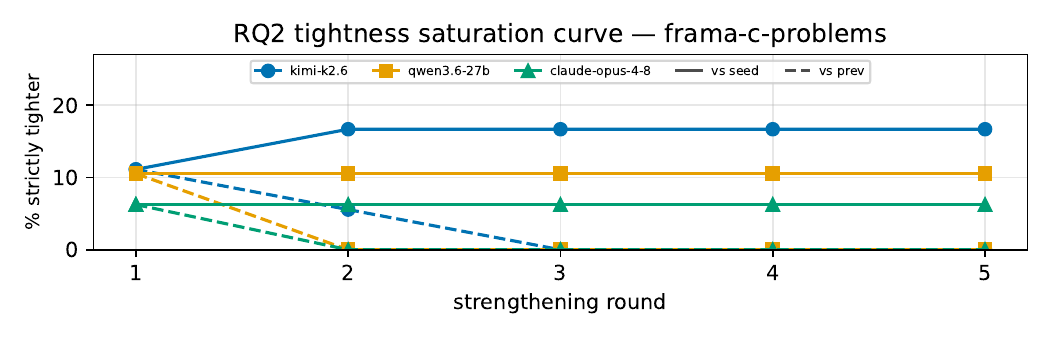}
  \caption{RQ2 tightness saturation on Frama-C-Problems. Per model, the share
    of orderable targets strictly tighter than the partial seed at round $k$
    (solid) and the round-over-round gain over the previous round (dashed). The
    denominator is fixed across rounds to the targets \textsc{judge} orders in
    every round (16 / 18 / 19 for Opus 4.8 / Kimi K2.6 / Qwen3.6-27B), so the
  solid curve is monotone by construction.}
  \label{fig:rq2-sat}
\end{figure}

\begin{table}[t]
  \centering
  \caption{RQ2 tightness endpoint, per model. ``Partial / Full'' are the
  per-arm pass counts (full arm of 50).}
  \label{tab:rq2}
  \begin{tabular}{lccc}
    \toprule
    Model & Sat.\ @0 & Equiv : tighter & Partial / Full \\
    \midrule
    Kimi K2.6     & 43/47 & 13 : 2 & 40 / 41 \\
    Opus 4.8      & 40/42 & 18 : 2 & 36 / 36 \\
    Qwen3.6-27B   & 38/42 & 14 : 4 & 36 / 36 \\
    \bottomrule
  \end{tabular}
\end{table}

\subsubsection{Soundness and reach on EqBench-C}

\textsc{regver} runs the contract-free regression harness on all 272 pairs,
producing a verdict on 70; the rest time out or hit a tool error and are held
out. On the decided pairs it never fabricates an equivalence: 100\% recall,
zero false negatives (Figure~\ref{fig:eqbench}). Some of its false positives
are not over-detection but dataset mislabels, pairs EqBench labels \texttt{Eq}
that are actually non-equivalent; we found nine and filed them upstream
\ifanonymized{}
\textcolor{red}{[REDACTED FOR REVIEW]}
\else{}
\cite{eqbenchmislabels}
\fi, and independent work reports EqBench mislabels as well
\cite{sarker2026quantitativesymbolicpatchimpact}. We draw no cross-tool
ranking (see \S\ref{subsec:baselines}): the comparable checkers decide partial
rather than safety-preserving conditional equivalence, so EqBench-C serves as
a soundness check, not a competitive baseline.

\paragraph{Reach with contracts.} A developer runs the contract-free
baseline first and escalates only the pairs it cannot decide; on that
undecided tail \textsc{scribe} infers a one-sided contract for the
changed function and \textsc{regver} runs again. This is where
contracts earn their keep: they recover 9, 11, and 7 correct decisions
for Opus, Kimi, and Qwen on pairs the baseline left as timeouts or
tool errors, pushing the decided count from 70 to 105 / 106 / 102
(Table~\ref{tab:eqbench-reach}), and every recovered decision is a
real divergence the baseline was too weak to expose. The gain comes
at an expected precision cost: \textit{the inferred contracts
over-approximate, so the recovered tail is false-alarm-heavy}. But
each false alarm lands on a pair the baseline could not decide at
all, a triageable over-approximation rather than a regression.
Soundness holds under the relation we target: no model fabricates an
equivalence, the lone dataset-scored exception being
\texttt{triangularMod} under Qwen, non-equivalent under full
equivalence but equivalent under the safety-preserving conditional
equivalence \textsc{regver} targets, since the two versions agree
wherever the old one terminates; we report it as the dataset scores
it rather than substitute our own oracle. Contracts therefore extend
reach into EqBench-C's hardest cases and fail safe: every decision
they add is correct or a triageable false alarm, never a missed difference.

\begin{table}[t]
  \centering
  \caption{EqBench-C reach. Counts are out of 272 pairs.}
  \label{tab:eqbench-reach}
  \begin{threeparttable}
    \begin{tabular}{lcccc}
      \toprule
      Arm & Decided & Correct & FP & FN \\
      \midrule
      Baseline (contract-free)  & 70  & 56 & 14 & 0 \\
      \midrule
      + contracts, Opus 4.8     & 105 & 65 & 40 & 0 \\
      + contracts, Kimi K2.6    & 106 & 67 & 39 & 0 \\
      + contracts, Qwen3.6-27B  & 102 & 63 & 38 & 1\tnote{a} \\
      \bottomrule
    \end{tabular}
    \begin{tablenotes}
      \footnotesize
    \item[a] \texttt{triangularMod}: false FN under safety-preserving
      conditional equivalence.
    \end{tablenotes}
  \end{threeparttable}
\end{table}

\begin{figure}[t]
  \centering
  \includegraphics[width=\columnwidth]{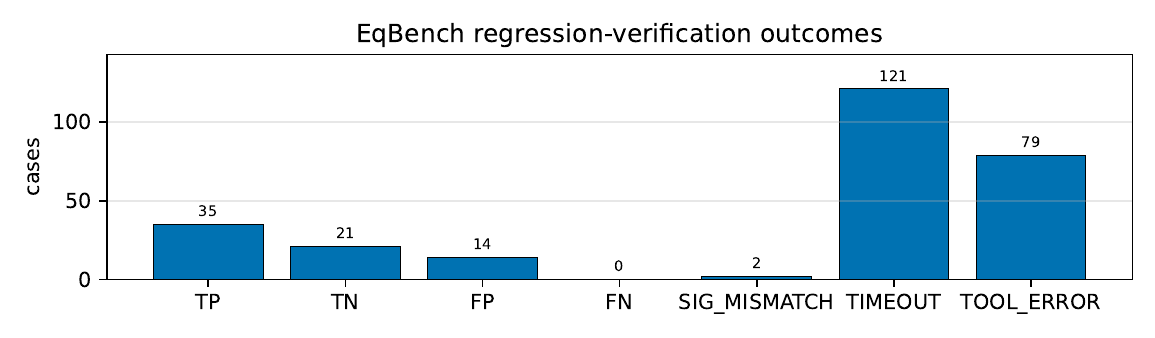}
  \caption{Contract-free baseline regression-verification outcomes on
    EqBench-C. The false-negative bucket (a non-equivalent pair called
  equivalent) is empty: the baseline never fabricates an equivalence.}
  \label{fig:eqbench}
\end{figure}

\subsection{Discussion}

The results cohere around the contributions. RQ1 supports C2, that BMC trace
feedback is enough to drive frame inference: a single model checker driven by
counterexamples infers verifiable contracts competitively with heavier
weakest-precondition pipelines. RQ2 supports the partial-spec stopping point
of Theorem~\ref{thm:t1} (C1), the cheap partial contract already close to the
tightest a strengthening loop reaches, with \textsc{judge} (C3) supplying the
tightness ordering behind that claim and reporting unorderable rounds as a
scope bound. The refinement telemetry shows the mechanism, enforcement the
binding gate and the loop only ever weakening an over-strong postcondition;
and the EqBench-C baseline backs the soundness half of C1 directly, with zero
fabricated equivalences.

\subsection{Threats To Validity}\label{sec:limitations}
\textbf{Construct validity.} The cross-tool RQ1 comparison spans different
success oracles, so we present it as a situating comparison rather than a
head-to-head. On EqBench-C \textsc{regver} runs in two arms against
the dataset's
full-equivalence labels as a soundness check, not a tool ranking; the reported
cascade composes them post hoc, faithfully, since pairs are scored
independently and the escalation gate reads only the baseline's abstention, not
the label. Because \textsc{regver} targets the weaker conditional equivalence, a
pair the dataset scores as a miss may be correct under our notion
(\texttt{triangularMod}); scoring against the dataset regardless, such
mismatches can only count against us. \textsc{judge} (C3)
reports a relative tightness ordering on
admit sets, not whether a contract characterises $f$ completely; we always
pair its verdict with enforcement, so a tighter-under-\textsc{judge} contract
that also enforces is unambiguously a tighter \emph{sound} contract. The C1
guarantee holds under its assumptions on \texttt{assigns} aliasing,
nondeterministic external calls, and termination; outside those it degrades
to an honest unknown rather than a false proof.

\textbf{Internal validity.} \textsc{scribe} runs each model once at
temperature $0$, a single deterministic trajectory rather than an average
over LLM sampling; we mitigate by reporting three models whose refinement
telemetry is near-identical, but a different seed could shift individual
cases. Non-results are excluded from the pass-rate denominator, which the
two-way reporting (\S\ref{sec:eval}) keeps honest since none is a wrong
verdict. On EqBench-C the size and make-up of the undecided tail the treatment
arm escalates on depend on the per-query timeout, fixed at 300s across both
arms. RQ2's warm start tightens only within the seed's clause shape and
cannot reach a structurally different contract; we frame it as the marginal
value of continued strengthening and keep \textsc{judge} strictly post hoc,
so the loop is gated on enforcement, never on the tightness measure it
reports.

\textbf{External validity.} The baseline comparison is one-way (published
numbers, not re-runs), on the same full-denominator basis they report. The
benchmarks are small (Frama-C-Problems averages about seventeen lines, the
X509 study six functions, EqBench-C crafted rather than production code), so
generalisation to large real-world C, and to C++, is future work. The
supported program class is bounded by engineering gaps with known fixes: the
tree-sitter call/include graph misses indirect, virtual, and macro-expanded
calls (an LLVM-IR pipeline with SVF~\cite{suixue2016svf} or libclang closes
this), the \textsc{regver} rename pass cannot rewrite identifiers in macro
bodies (Coccinelle~\cite{padioleau2008coccinelle}), and its side-effect
comparator sees only writes through pointer parameters
(bi-abduction~\cite{calcagno2009biabduction} and Infer-style footprints
extend it). \textsc{judge}'s tightness figures exist only on its
primitive-typed signature class. These are engineering scope bounds, not open
research questions.

\section{Conclusion}\label{sec:conclusion}

We asked whether a partial, caller-sufficient contract, rather than a
full behavioral specification, is enough to verify regression
equivalence. Theorem~\ref{thm:t1} answers yes for safety-preserving
conditional equivalence: a one-sided contract that abstracts only the
old version, scoped to the property worth preserving, is sound for the
regression check, and a depth-bounded closure (Lemma~\ref{lem:depth})
bounds the contract closure so the proof scales. The target of
inference is therefore not a
complete specification but the weakest contract a caller needs.

An end-to-end system (\textsc{scribe}, \textsc{contractor},
\textsc{regver}) infers these contracts automatically from the model
checker's own counterexamples, with no separate specification step. On
Frama-C-Problems and the ANSSI X509 parser this reaches a verification
rate comparable to the weakest-precondition tools AutoSpec and Preguss
while certifying a stronger property, enforcement plus
caller-sufficiency with runtime-error freedom included, which answers
RQ1: counterexample feedback alone is a sufficient signal for frame
inference. Strengthening each inferred contract past caller-sufficiency
barely tightens it, the partial and warm-started contracts coming out
\textsc{judge}-equivalent for the large majority of comparable targets,
so the partial-spec stopping point of Theorem~\ref{thm:t1} is close to
the practical ceiling rather than a compromise (RQ2).

A decision procedure for contract subsumption (\textsc{judge}) reports
these tightness results as a real set-inclusion ordering rather than a
gameable clause count. The regression check underneath is sound by
construction: on EqBench-C \textsc{regver} never fabricates an
equivalence, and run as a soundness oracle it surfaced nine pairs the
suite mislabels as equivalent, more than concurrent work reports.
Together these establish that sound, modular regression verification
can be driven by a model checker's own counterexamples and a partial
contract, with no hand-written specification.

Two directions follow naturally. \textbf{LLM-driven C++ contract
inference:} every published LLM contract-inference tool we are aware of
is C-only, so inferring functional contracts for C++ operator
overloads, initializer-list constructors, RAII destructors, and
template instantiations over the standard library is an unaddressed
problem and the natural next target for this line of work.
\textbf{Regression verification across signature changes:}
RVT~\cite{godlin2013rv}, Reve~\cite{klebanov2018reve},
SymDiff~\cite{lahiri2012symdiff}, and \textsc{regver} all assume the
old and new versions share a signature, so aligning refactored
signatures (parameter reordering, splitting, merging, type changes) for
verification-grade equivalence is an open problem, naturally cast as a
relational extension of Theorem~\ref{thm:t1} in which the parameter
alignment is itself part of the inference.

\section{Acknowledgements}

The authors acknowledge the use of generative AI tools in both the
development of the research tooling and the preparation of this
manuscript. All AI-assisted content was subsequently reviewed, revised,
and validated by the authors, who assume full responsibility for the
integrity and accuracy of the final publication.

\clearpage

\bibliographystyle{IEEEtran}
\bibliography{refs}

\end{document}